\documentclass[lettersize,journal]{IEEEtran}

\usepackage[T1]{fontenc}% optional T1 font encoding
\usepackage{amssymb}
\usepackage[utf8]{inputenc}
\usepackage[english]{babel}
\usepackage{amsmath}				
\usepackage{comment}				
\usepackage[T1]{fontenc}
\usepackage{mathtools}
\usepackage{algcompatible,amsmath}
\usepackage[ruled,vlined]{algorithm2e}
\usepackage{graphicx}
\usepackage{multicol}
\usepackage{caption}
\usepackage{afterpage}
\usepackage[noend]{algpseudocode}
\usepackage{color}
\usepackage{etoolbox}
\usepackage{url}
\usepackage{siunitx,booktabs}
\usepackage{placeins}
\usepackage{bbm}
\usepackage{soul}
\usepackage[noadjust]{cite}

\usepackage{amsfonts}

\usepackage{accents}

\newcommand\myeq{\stackrel{\mathclap{\normalfont\mbox{(a)}}}{=}}

\normalsize
\makeatletter
\def\underbracex#1#2{\mathop{\vtop{\m@th\ialign{##\crcr
				$\hfil\displaystyle{#2}\hfil$\crcr
				\noalign{\kern3\p@\nointerlineskip}%
				#1\crcr\noalign{\kern3\p@}}}}\limits}

\def\upbracefilla{$\m@th \setbox\z@\hbox{$\braceld$}%
	\bracelu\leaders\vrule \@height\ht\z@ \@depth\z@\hfill 
	\kern\p@\vrule \@width\p@\kern\p@\vrule \@width\p@\kern\p@\vrule \@width\p@
	$}

\def\upbracefillb{$\m@th \setbox\z@\hbox{$\braceld$}%
	\vrule \@width\p@\kern\p@\vrule \@width\p@\kern\p@\vrule \@width\p@\kern\p@
	\leaders\vrule \@height\ht\z@ \@depth\z@\hfill\bracerd
	\braceld\leaders\vrule \@height\ht\z@ \@depth\z@\hfill
	\kern\p@\vrule \@width\p@\kern\p@\vrule \@width\p@\kern\p@\vrule \@width\p@
	$}

\def\upbracefillc{$\m@th \setbox\z@\hbox{$\braceld$}%
	\vrule \@width\p@\kern\p@\vrule \@width\p@\kern\p@\vrule \@width\p@\kern\p@
	\leaders\vrule \@height\ht\z@ \@depth\z@\hfill
	\kern\p@\vrule \@width\p@\kern\p@\vrule \@width\p@\kern\p@\vrule \@width\p@
	$}

\def\upbracefilld{$\m@th \setbox\z@\hbox{$\braceld$}%
	\vrule \@width\p@\kern\p@\vrule \@width\p@\kern\p@\vrule \@width\p@\kern\p@
	\leaders\vrule \@height\ht\z@ \@depth\z@\hfill\braceru$}

\makeatletter
\patchcmd{\@makecaption}
{\scshape}
{}
{} 
{}
\makeatother
\usepackage{subcaption}
\usepackage{graphicx}
\usepackage{amsthm}

\newtheorem{proposition}{Proposition}

\newtheorem*{proof*}{Proof}

\newcommand{\norm}[1]{\left\lVert#1\right\rVert}

\IEEEoverridecommandlockouts
\makeatletter
\let\origIEEEPARstart\IEEEPARstart
\renewcommand{\IEEEPARstart}[3][1.1]{%
	\def\@IEEEPARstartDROPDEPTH{#1\baselineskip}%
	\origIEEEPARstart{#2}{#3}%
}

\def\endthebibliography{%
	\def\@noitemerr{\@latex@warning{Empty `thebibliography' environment}}%
	\endlist
}
\makeatother
\ifCLASSINFOpdf
\else
\fi
\begin{document}
	\title{Precoding for Uplink RIS-Assisted Cell-Free MIMO-OFDM Systems with Hardware Impairments}
	\vspace{-1cm}
	\author{Navid~Reyhanian,
	Reza Ghaderi Zefreh, Parisa Ramezani, and Emil Bj{\"o}rnson,
	\IEEEmembership{Fellow, IEEE}
	\thanks{N. Reyhanian was with the Department of Electrical Engineering, University of Minnesota, Minneapolis, MN, USA 55455. He is now with Cisco Systems, Milpitas, CA, USA 95035  (email: navid@umn.edu).}% <-this % stops a space
%	\thanks{N. Reyhanian is with Cisco Systems, Milpitas, CA, USA 95035, also with the Department of Electrical Engineering, University of Minnesota, Minneapolis, MN, USA 55455 (email: nreyhani@cisco.com).}% <-this % stops a space
	\thanks{R.G. Zefreh is Independent Researcher, Shiraz, Fars, Iran (email: r.ghaderi@alumni.iut.ac.ir).}% <-this % stops a space
	\thanks{P. Ramezani and E. Bj{\"o}rnson are with the Department of Computer Science, KTH Royal Institute of Technology, Stockholm, Sweden (email: \{parram,emilbjo\}@kth.se).}\vspace{-0.7cm}
}
	
	\allowdisplaybreaks
	
	\maketitle
	
	\begin{abstract}
		This paper studies a reconfigurable intelligent surface (RIS)-assisted cell-free massive multiple-input multiple-output (CF-mMIMO) system with multiple RISs. Joint design of transmit precoding, RIS coefficients, and receive combining is investigated for uplink sum-rate maximization under in-phase and quadrature phase imbalance (IQI) at user equipments (UEs) and access points (APs). A weighted minimum mean squared error (WMMSE) based block coordinate descent (BCD) approach is proposed, where novel iterative methods are developed to efficiently solve the BCD subproblems. The efficiency of proposed approaches is demonstrated relative to classical and heuristic methods via extensive simulations.
	\end{abstract}
	
	\IEEEpeerreviewmaketitle
	\section{Introduction}
	\IEEEPARstart{C}{ell-free} massive multiple-input multiple-output (CF-mMIMO) is a key technology for the next generation wireless networks as it combines the benefits of mMIMO and ultra-dense networks \cite{10071534}, while mitigating their drawbacks. The main feature of CF-mMIMO is that each user equipment (UE) can be associated with several geographically distributed access-points (APs) and they can simultaneously serve each UE through coherent transmission. In certain variants of CF-mMIMO networks that feature a central processing unit (CPU), fronthaul links facilitate the connection of APs to the CPU through either high-capacity wired or wireless connections. The CPU is responsible for the UE-AP coordination and data processing, enabling advanced precoding methods that achieve higher spectral efficiency (SE) than non-coordinated techniques like maximum ratio transmission (MRT) \cite{8901451}. 
	
Realizing high SE usually demands significant infrastructure. CF-mMIMO systems assisted by low-cost reconfigurable intelligent surface (RIS) techniques are a promising solution, where joint AP beamfoming and RIS reflection design can bring substantial sum rate gains \cite{9656609,9392378,9530675}. In modern wideband systems, low-cost hardware in UEs and APs can introduce in-phase (I) and quadrature-phase (Q) signal imbalances, which can be particularly detrimental to orthogonal frequency-division multiplexing (OFDM) performance. These I/Q imbalances (IQI) complicate precoding and combining processes, making robust mitigation algorithms essential.
	
	Precoding in the presence of IQI has been recently studied \cite{9668928,10000842,10044153,9348619,11114797}. In \cite{9668928}, the hybrid beamforming and its performance in downlink multi-user systems with receiver IQI are studied. The authors of \cite{10000842} proposed a precoding scheme for uplink CF-mMIMO systems with IQI and low-resolution analog-to-digital  converters (ADCs) at receivers. This approach involves  a two-layer decoding and MRT precoding.  A similar receiver IQI issue is addressed in \cite{10044153} in the presence of channel sounding errors arising from the movements of single-antenna UEs in uplink transmissions.  Moreover, \cite{7577860} develops uplink baseband algorithms for mMIMO systems to jointly handle signal distortions from low-precision IQ modulators and low-resolution ADCs by proposing a two-stage IQ imbalance estimation method and new channel estimators and detectors. One limitation of the above-mentioned papers is that they assume ideal transmitters and only consider IQI at receivers. A deep neural network (DNN) based downlink single-carrier hybrid beamforming with different hardware impairments such as IQI and power amplifier nonlinearities is studied in \cite{11114797} for single-user MIMO systems, where the goal is to minimize the mean squared distortion error of transmitted symbols. Due to the black-box nature of the DNN approach in \cite{11114797}, it remains challenging to predict whether or not the proposed method generalizes to the multi-user MIMO-OFDM scenario.

	This paper investigates the maximization of the uplink sum-rate in fully centralized RIS-assisted multi-user CF-mMIMO-OFDM systems. We propose a novel approach involving joint optimization of UE precoding, multi-RIS coefficients, and combining at APs, while considering IQI impairments at both UEs and APs.
	The formulated problem is non-convex and addressed using a novel block coordinate descent (BCD) algorithm based on the weighted minimum mean square error (WMMSE) method. This BCD approach effectively solves subproblems related to precoders at UEs and combiners at APs through bisection-search methods and closed-form solutions, respectively.
	Several reformulations of the subproblem for multi-RIS coefficients are proposed. These reformulations are globally solved using a novel, cost-effective BCD method with closed-form subproblem update rules. The convergence of the proposed methods is theoretically supported, and extensive numerical tests are given to demonstrate their superiority over the classical MMSE method, heuristic methods that solve disjoint problems, and methods that neglect hardware impairments.

	\section{System Model and Problem Formulation}
	The uplink of a CF-mMIMO system is considered. A total of $C$ APs are considered, each equipped with $N_r$ receive antennas. $K$ UEs, each with $N_t$ transmit antennas, are connected to APs through direct uplink channels and potentially through multiple (here $Q$) passive RISs.
	
	Suppose that the available spectrum is divided into $S$ subcarriers where each UE is served through a subset of subcarriers collected in $\mathcal{S}_k$. The set of UEs being served via the  $s^\text{th}$ subcarrier is denoted by $\mathcal{K}_{s}$. The channel between UE $k$ and the $j^{\text{th}}$ RIS, equipped with $M$ elements, on subcarrier $s\in\{-S/2,\dots, -1, 1,\dots, S/2\}$ is denoted by $\mathbf{H}_{j,k}^s\in \mathbb{C}^{ M\times N_t}$. The channel between the $j^{\text{th}}$ RIS and the $c^{\text{th}}$ AP on the $s^\text{th}$ subcarrier is represented by $\mathbf{G}_{c,j}^s\in \mathbb{C}^{ N_r\times M}$. The overall channel from $Q$ RISs to  the $c^{\text{th}}$ AP on the $s^\text{th}$ subcarrier is denoted by $\mathbf{G}_{c}^s= [\mathbf{G}_{c,1}^s, \dots,\mathbf{G}_{c,Q}^s]\in\mathbb{C}^{N_r\times QM}$. Let $\mathbf{R}_{c,k}^s\in \mathbb{C}^{ N_r\times N_t}$ denote the direct uplink channel from UE $k$ to the $c^{\text{th}}$ AP on the $s^\text{th}$ subcarrier. Moreover, we define the concatenated channel to all APs as $\mathbf{R}_{k}^s=[\mathbf{R}_{1,k}^{s^T},\dots,\mathbf{R}_{C,k}^{s^T}]^T\in \mathbb{C}^{ CN_r\times N_t}$.
	
	The amplitude and phase shift of the coefficient for the $m^{\text{th}}$ element of the $j^{\text{th}}$ RIS are represented by $\rho_m^j$ and $\phi_m^j$, respectively, where  $0\leq \rho_m^j\leq 1$ and $-\pi\leq \phi_m^j\leq \pi$ \cite{9459505}. Let $\boldsymbol{\Theta}^j$ denote a diagonal coefficient matrix of the $j^{\text{th}}$ RIS, where the $m^\text{th}$ diagonal element is $[\boldsymbol{\Theta}^j]_{m,m}=\rho_{m}^je^{\iota\phi_m^j}$ ($\iota =\sqrt{-1}$). We assume that each RIS element identically applies phase rotations and amplifications  to all subcarriers \cite{9133184}.
	
	Let $b_k$ denote the number of spatial streams of the $k^{\text{th}}$ UE. The intended symbol vector to UE $k$ is circularly symmetric complex Gaussian variable and denoted by $\mathbf{x}_k^s\in\mathbb{C}^{b_k\times 1}$, satisfying $\mathbb{E}[\mathbf{x}_k^s\mathbf{x}_k^{s^H}]=\mathbf{I}_{b_k\times b_k}$ and $\mathbb{E}[\mathbf{x}_k^s\mathbf{x}_l^{s^H}]=\mathbf{0}_{b_k\times b_k}$, for $l\neq k$. The precoder matrix for UE $k$ on the $s^\text{th}$ subcarrier is denoted by $\mathbf{V}_k^s\in\mathbb{C}^{N_t\times b_k}$. The received signal by the APs on the $s^\text{th}$ subcarrier with ideal hardware is given by concatenating the signal $\mathbf{y}_c^s$ (received by the $c^{\text{th}}$ AP) as $\mathbf{y}^s = (\mathbf{G}^s\boldsymbol{\Theta}\mathbf{H}_{k}^s+\mathbf{R}_k^s)\mathbf{V}_k^s\mathbf{x}_k^s+\sum_{i=1,i\neq k}^{K} (\mathbf{G}^s\boldsymbol{\Theta}\mathbf{H}_{i}^s+\mathbf{R}_i^s)\mathbf{V}_i^s\mathbf{x}_i^s+\mathbf{n}^s$,
	where $\mathbf{y}^s=[\mathbf{y}_{1}^{s^T},\dots,\mathbf{y}_{C}^{s^T}]^T\in\mathbb{C}^{CN_r\times 1}$, $\mathbf{G}^s= [\mathbf{G}_{1}^{s^T},\dots,\mathbf{G}_{C}^{s^T}]^T\in\mathbb{C}^{CN_r\times QM}$, $\boldsymbol{\Theta}=\text{diag}(\{\boldsymbol{\Theta}^j\}_{j=1}^Q)$, $\mathbf{H}_{k}^s=[\mathbf{H}_{1,k}^{s^T},\dots,\mathbf{H}_{Q,k}^{s^T}]^T\in\mathbb{C}^{QM\times N_t}$, and $\mathbf{n}^s=[\mathbf{n}_{1}^{s^T},\dots,\mathbf{n}_{C}^{s^T}]^T\in\mathbb{C}^{CN_r\times 1}$. We consider $\mathbf{n}^s\sim\mathcal{CN}(\mathbf{0},\sigma^2\mathbf{I}_{CN_r\times CN_r})$. The summation term represents the undesired interference for UE $k$ signal.  Signals that are reflected by two or more RISs are ignored, as multi-hop links face large path loss and undergo significant attenuation. Uplink transmissions by each UE on the $s^\text{th}$ subcarrier is limited under a power constraint as $\norm{\mathbf{V}_k^s}_F^2\leq p_k^s$, where $p_k^s$ is the  power budget for UE $k$. 
	
	For simplicity, let $\bar{\mathbf{H}}_k^s$ represent the effective channel from UE $k$ to APs, which equals to $\mathbf{G}^s\boldsymbol{\Theta}\mathbf{H}_{k}^s+\mathbf{R}_k^s$. We assume that the primary source of impairments is IQI at APs and UEs. Here, the received signal at APs from UE $k$ on the $s^\text{th}$ subcarrier, with impairments at both ends, is described as \cite{4024126}:
	\begin{align}
		\mathbf{y}^s =& \mathbf{P}_k^{1,s}\mathbf{V}_k^{s}\mathbf{x}_k^s+\sum_{i\in\mathcal{K}_s,i\neq k} \mathbf{P}_i^{1,s}\mathbf{V}_i^{s}\mathbf{x}_i^s+\sum_{i\in\mathcal{K}_s} \mathbf{P}_i^{2,s}\mathbf{V}_i^{-s^*}\mathbf{x}_i^{-s^*}\nonumber\\
		& +\mathbf{K}_1\mathbf{n}^s+\mathbf{K}_2\mathbf{n}^{s^*},   \label{eq:rec}
	\end{align}
where $\mathbf{P}_k^{1,s} = \mathbf{K}_1\bar{\mathbf{H}}_k^s\mathbf{D}_{1k}+\mathbf{K}_2\bar{\mathbf{H}}_k^{-s^*}\mathbf{D}_{2k}$, $\mathbf{P}_k^{2,s} = \mathbf{K}_2\bar{\mathbf{H}}_k^{{-s}^*}\mathbf{D}_{1k}^{*}+\mathbf{K}_1\bar{\mathbf{H}}_k^{s}\mathbf{D}_{2k}^*$.
In the above equations, $\mathbf{K}_{1c}=(\mathbf{I}_{N_r\times N_r}+\mathbf{L}_{r,c}\odot\text{diag}(e^{-\iota\boldsymbol{\psi}_{r,c}}))/2$, $\mathbf{K}_1=\text{diag}(\{\mathbf{K}_{1c}\}_{c=1}^C)$, $\mathbf{K}_2=\mathbf{I}_{CN_r\times CN_r}-\mathbf{K}_1^*$, $\mathbf{D}_{1k}=(\mathbf{I}_{N_t\times N_t}+\mathbf{L}_{t,k}\odot\text{diag}(e^{\iota\boldsymbol{\psi}_{t,k}}))/2$, $\mathbf{D}_{2k}=\mathbf{I}_{N_t\times N_t}-\mathbf{D}_{1k}^*$, where $\mathbf{L}_{t,k}$ and $\mathbf{L}_{r,c}$ denote real diagonal matrices representing amplitude mismatches at the $k^\text{th}$ transmitter and the $c^\text{th}$ receiver, respectively, the vectors $\boldsymbol{\psi}_{t,k}$ and $\boldsymbol{\psi}_{r,c}$ model the phase mismatches \cite{4024126}. For ideal hardware, we have $\mathbf{L}_{t,k}=\mathbf{I}_{N_t\times N_t}$, $\mathbf{L}_{r,c}=\mathbf{I}_{N_r\times N_r}$, $\boldsymbol{\psi}_{t,k}=\mathbf{0}$, and $\boldsymbol{\psi}_{r,c}=\mathbf{0}$. Furthermore, $\odot$ represents the Hadamard product.
	From \eqref{eq:rec}, it is observed that the received signal on the $s^\text{th}$ subcarrier is affected by cross-talk from signals transmitted over the subcarrier $-s$. This phenomenon complicates the precoding process, as the precoder $\mathbf{V}_k^{s}$ must not only maximize data transmissions for UE $k$ on subcarrier $s$ but also minimize the cross-talk on subcarrier $-s$.
	
	Once the signals from the $k^{\text{th}}$ UE are received at the $c^{\text{th}}$ AP, a combiner matrix $\mathbf{U}_{c,k}^s\in\mathbb{C}^{N_r\times b_k}$ is applied \cite{9737367}. Subsequently, the locally decoded data streams of the UE are routed through fronthaul links to a CPU, where channel estimates, impairments, precoding matrices, and RIS coefficients are available. To simplify the receiver design, $\mathbf{x}_k^s$ is estimated from subcarrier $s$, and the cross-talk term $\mathbf{x}_k^{s^*}$ included in $\mathbf{y}^{-s}$ is treated as uncorrelated noise \cite{10044153}. In the model considered, the CPU recovers the spatial streams of each UE as $\hat{\mathbf{x}}_k^s=\mathbf{U}_{k}^{s^H} \mathbf{y}^s$, where $\mathbf{U}_{k}^s=[\mathbf{U}_{1,k}^{s^T},\dots,\mathbf{U}_{C,k}^{s^T}]^T\in\mathbb{C}^{CN_r\times b_k}$. 
	
	\begin{proposition}
		When signals of interfering UEs and the cross-talk from subcarrier $-s$ are treated as uncorrelated noise, the achievable rate with a fully centralized processing is
		\begin{align}
			\text{SE}_k^s=\log_2\left|\mathbf{I}_{b_k\times b_k}+\mathbf{V}_k^{s^H}\mathbf{P}_k^{1,s^H}  \mathbf{J}_k^{s^{-1}}\mathbf{P}_k^{1,s}\mathbf{V}_k^s\right|,\label{eq:capacity}
		\end{align}
		where $\mathbf{J}_k^s=\sum_{i\in\mathcal{K}_s, i\neq k}  \mathbf{P}_i^{1,s}\mathbf{V}_{i}^s\mathbf{V}_{i}^{s^H}\mathbf{P}_i^{{1,s}^H}+\sigma^2\mathbf{K}_1\mathbf{K}_1^H+\sigma^2\mathbf{K}_2\mathbf{K}_2^H+\sum_{i\in\mathcal{K}_{s}}\mathbf{P}_i^{2,s}\mathbf{V}_i^{-s^*}\mathbf{V}_i^{-s^T}\mathbf{P}_i^{2,s^H}.$
	\end{proposition}
	\begin{proof}
		The mutual information for UE $k$ on subcarrier $s$ in the presence of cross-talk and interference is as:
		\begin{align}
			&I(\mathbf{x}_{k}^s;\mathbf{y}^s|\{\{\bar{\mathbf{H}}_k^s\}_{s\in\mathcal{S}_k}\}_{k=1}^K)\nonumber\\
			&=h(\mathbf{y}^s|\{\{\bar{\mathbf{H}}_k^s\}_{s\in\mathcal{S}_k}\}_{k=1}^K)-h(\mathbf{y}^s|\mathbf{x}_{k}^s,\{\{\bar{\mathbf{H}}_k^s\}_{s\in\mathcal{S}_k}\}_{k=1}^K)\nonumber\\
			&\geq \frac{1}{2}\log_2((2\pi e)^{CN_r}|\mathbf{P}_k^{1,s}\mathbf{V}_{k}^s\mathbf{V}_{k}^{s^H}\mathbf{P}_k^{{1,s}^H}+\mathbf{J}_k^s|)\nonumber\\
			&-\frac{1}{2}\log_2((2\pi e)^{CN_r}|\mathbf{J}_k^s|)\myeq\log_2\left|\mathbf{I}_{b_k\times b_k}\hspace{-.1cm}+\hspace{-.1cm}\mathbf{V}_k^{s^H}\mathbf{P}_k^{1,s^H}  \mathbf{J}_k^{s^{-1}}\mathbf{P}_k^{1,s}\mathbf{V}_k^s\right|\nonumber
		\end{align}
		where (a) follows from the matrix determinant lemma. This  lower-bound is used as the SE for UE $k$ on subcarrier $s$.
	\end{proof}
	The goal of this paper is to maximize the overall system sum-rate by jointly optimizing the precoding matrices for UEs, and  RIS coefficients as
	\begin{subequations}\label{opt:main}
		\begin{align}
			\max_{\boldsymbol{\Theta},\{\{\mathbf{V}_k^s\}_{s\in\mathcal{S}_k}\}_{k=1}^K} \quad & \sum_{k=1}^{K}\sum_{s\in \mathcal{S}_k}\text{SE}_k^s\label{eq:spece}\\
			\textrm{s.t.} \quad & \norm{\mathbf{V}_k^s}_F^2 \leq p_k^s, \hspace{.3cm}s\in\mathcal{S}_k,\forall k,\label{eq:power}\\
			& |[\boldsymbol{\Theta}^j]_{m,m}|\leq 1, \hspace{.2cm}\forall m, \forall j.\label{eq:ampin}
		\end{align}
	\end{subequations}
	In the above problem, $|[\boldsymbol{\Theta}^j]_{m,m}|\leq 1$ ensures that the reflected signal from each element is not amplified.
	
	\section{The Proposed WMMSE-Based BCD Approach}
	Here, we concentrate on the relationship between the total sum-rate and the mean-square error (MSE) of the optimal combining matrices to reformulate \eqref{opt:main}. We propose a WMMSE-based approach, wherein we minimize a weighted MSE with respect to the aforementioned variables, rather than solving \eqref{opt:main}. First, we derive the MSE matrix for UE $k$  on subcarrier $s$ as below:
	\begin{align}
		&\mathbf{E}_{k}^s =\mathbb{E}[(\hat{\mathbf{x}}_k^s-\mathbf{x}_k^s)(\hat{\mathbf{x}}_k^s-\mathbf{x}_k^s)^H]=\mathbf{I}_{b_k\times b_k}\hspace{-.1cm}-\mathbf{U}_{k}^{s^H}\mathbf{P}_k^{1,s}\mathbf{V}_{k}^s\label{eq:E}\\
		&-\mathbf{V}_{k}^{s^H}\mathbf{P}_k^{1,s^H}\mathbf{U}_{k}^s+\mathbf{U}_{k}^{s^H}\mathbf{P}_k^{1,s}\mathbf{V}_k^s\mathbf{V}_k^{s^H}\mathbf{P}_k^{1,s^H}\mathbf{U}_{k}^{s}+\mathbf{U}_{k}^{s^H}\mathbf{J}_k^s\mathbf{U}_{k}^s.\nonumber
	\end{align}
	The above expectation is with respect to $\{\mathbf{x}_k^s\}_{k=1}^K$ and $\mathbf{n}^s$. Based on the WMMSE approach \cite{5756489}, instead of solving \eqref{opt:main}, a positive semi-definite (PSD) weight matrix $\mathbf{W}_{k}^s$ is applied to $\mathbf{E}_{k}^s$, and the following problem is formulated:
\begin{subequations}\label{opt:wmmse_v3}
	\begin{align}
		\min_{\substack{\boldsymbol{\Theta},\{\{\mathbf{U}_k^s,\mathbf{V}_k^s, \\ \mathbf{W}_k^s \succeq \mathbf{0}\}_{s\in \mathcal{S}_k}\}_{k=1}^K}} & \sum_{k=1}^{K}\sum_{s\in \mathcal{S}_k} \! \mathrm{Tr}(\mathbf{W}_{k}^s\mathbf{E}_{k}^s)\!-\!\log_2(|\mathbf{W}_{k}^s|)\label{eq:wmmse_v3} \\[-1.5ex] %<-- Reduces vertical space
		\textrm{s.t.} \quad & \eqref{eq:power}-\eqref{eq:ampin}.
	\end{align}
\end{subequations}
	To solve the above problem, we minimize \eqref{opt:wmmse_v3} alternatively with respect to different blocks via a BCD approach.

	\subsection{The Subproblem with Respect to $\mathbf{U}_k^s$} 
	The problem with respect to $\mathbf{U}_k^s$ is unconstrained. The optimal solution for $\mathbf{U}_k^s$ is determined by the first order optimality condition. Hence, each $\mathbf{U}_k^s$ iterate is updated as $\mathbf{U}_k^s = \Big( \mathbf{P}_k^{1,s}\mathbf{V}_k^s\mathbf{V}_k^{s^H}\mathbf{P}_k^{{1,s}^H}+\mathbf{J}_k^s\Big)^{-1}\mathbf{P}_k^{1,s}\mathbf{V}_k^s$.
	
	\subsection{The Subproblem with Respect to $\mathbf{W}_k^s$}
Replacing the optimal $\mathbf{U}_k^s$ in \eqref{eq:E}, the MSE is obtained as $\mathbf{E}_{k}^{s,o} =\mathbf{I}_{b_k\times b_k}-\mathbf{V}_{k}^{s^H}\mathbf{P}_k^{1,s^H}\Big( \mathbf{P}_k^{1,s}\mathbf{V}_k^s\mathbf{V}_k^{s^H}\mathbf{P}_k^{{1,s}^H}+\mathbf{J}_k^s\Big)^{-1}\mathbf{P}_k^{1,s}\mathbf{V}_k^s$.
	We substitute $\mathbf{E}_{k}^{s,o}$ into \eqref{eq:wmmse_v3} and obtain $\mathbf{W}_{k}^s$ from the first order optimality condition as $\mathbf{W}_k^s=(\mathbf{E}_{k}^{s,o})^{-1} = (\mathbf{I}_{b_k\times b_k}-\mathbf{U}_k^{s^H}\mathbf{P}_k^{1,s}\mathbf{V}_k^{s})^{-1}$. \begin{comment}
		Using this value for  $\mathbf{W}_k^s$, we can re-calculate \eqref{eq:wmmse} as
		\begin{align}
			&\log_2\left|(\mathbf{E}_{k}^{s,o})^{-1}\right|=\log_2\left|(\mathbf{I}_{b_k\times b_k}-\mathbf{U}_k^{s^H}\mathbf{P}_k^{1,s}\mathbf{V}_k^s)^{-1}\right|\nonumber\\
			&=\log_2\left|\mathbf{I}_{b_k\times b_k}\hspace{-.1cm}-\hspace{-.1cm}\mathbf{V}_{k}^{s^H}\mathbf{P}_k^{1,s^H}\hspace{-.1cm}\Big( \mathbf{P}_k^{1,s}\mathbf{V}_k^s\mathbf{V}_k^{s^H}\mathbf{P}_k^{{1,s}^H}\hspace{-.2cm}+\mathbf{J}_k^s\Big)^{-1}\mathbf{P}_k^{1,s}\mathbf{V}_k^s\right|\nonumber\\
			&\myeq\log_2\left|\mathbf{I}_{b_k\times b_k}+\mathbf{V}_k^{s^H}\mathbf{P}_k^{1,s^H}  \mathbf{J}_k^{s^{-1}}\mathbf{P}_k^{1,s}\mathbf{V}_k^s\right|\nonumber,
		\end{align}
		where (a) follows from the Woodbury matrix identity. We note that the above expression is equal to \eqref{eq:capacity}. Therefore, minimizing \eqref{eq:wmmse} is equivalent to maximizing \eqref{eq:spece}.
		\vspace{-.2cm}
	\end{comment}
	\subsection{The Subproblem with Respect to $\mathbf{V}_k^s$}
Observe that \eqref{eq:wmmse_v3} is quadratic and convex with respect to $\mathbf{V}_k^s$. Furthermore, there is one maximum power budget constraint per-subcarrier per UE. Dualizing the power constraint, the Lagrangian is expressed as  $\mathcal{L}(\mathbf{V}_{k}^s, \vartheta_k^s) =\sum_{i\in\mathcal{K}_s }\text{Tr}(\mathbf{W}_i^s\mathbf{U}_{i}^{s^H}\mathbf{P}_k^{1,s}\mathbf{V}_k^s\mathbf{V}_k^{s^H}\mathbf{P}_k^{1,s^H}\mathbf{U}_{i}^s)+\vartheta_k^s(\norm{\mathbf{V}_k^s}_F^2 - p_k^s)- \text{Tr}(\mathbf{W}_{k}^s\mathbf{U}_{k}^{s^H}\mathbf{P}_k^{1,s}\mathbf{V}_{k}^s)-\text{Tr}(\mathbf{W}_{k}^s\mathbf{V}_{k}^{s^H}\mathbf{P}_k^{1,s^H}\mathbf{U}_{k}^s)+\sum_{i\in\mathcal{K}_{-s} }\text{Tr}(\mathbf{W}_i^{-s}\mathbf{U}_{i}^{-s^H}\mathbf{P}_k^{2,-s}\mathbf{V}_k^{s^*}\mathbf{V}_k^{s^T}\mathbf{P}_k^{2,-s^H}\mathbf{U}_{i}^{-s})$.
	One can obtain the minimizer of $\mathbf{V}_k^s$ using a bisection-search method on the Lagrange multiplier, $\vartheta_k^s$, associated with the power constraint over the positive orthant and having the KKT conditions satisfied:
	\begin{subequations}\label{eq:V}
		\begin{align}
			&\mathbf{V}_k^s = \Big(\sum_{i\in\mathcal{K}_s }\mathbf{P}_k^{1,s^H}\mathbf{U}_i^s\mathbf{W}_i^s\mathbf{U}_{i}^{s^H}\mathbf{P}_k^{1,s}+\sum_{i\in\mathcal{K}_{-s} }\mathbf{P}_k^{2,-s^T}\mathbf{U}_{i}^{-s^*}\nonumber\\
			&\mathbf{W}_{i}^{-s^T}\mathbf{U}_{i}^{-s^T}\mathbf{P}_k^{2,-s^*}+\vartheta_k^s\mathbf{I}_{N_t\times N_t}\Big)^{-1}\mathbf{P}_k^{1,s^H}\mathbf{U}_{k}^s\mathbf{W}_{k}^s,\label{eq:V1}\\
			&\norm{\mathbf{V}_k^s}_F^2 \leq p_k^s, \vartheta_k^s(\norm{\mathbf{V}_k^s}_F^2 - p_k^s)=0, \vartheta_k^s\geq 0. \label{eq:V2}
		\end{align}
	\end{subequations}
 If no positive $\vartheta_k^s$ satisfies the above condition, one can set $\vartheta_k^s=0$ and obtain $\mathbf{V}_k^s$ from \eqref{eq:V1}.
	\vspace{-.3cm}
	
	\allowdisplaybreaks

	\subsection{The Subproblem with Respect to $\boldsymbol{\Theta}$}
	This section proposes a scalable and low-cost method to minimize the objective function $\sum_{k=1}^{K}\sum_{s\in \mathcal{S}_k}\text{Tr}(\mathbf{W}_{k}^s\mathbf{E}_{k}^s)$ with respect to $\boldsymbol{\Theta}$, while guaranteeing global optimality. 
	Using \eqref{eq:E}, we expand $\sum_{k=1}^{K}\sum_{s\in \mathcal{S}_k}\text{Tr}(\mathbf{W}_{k}^s\mathbf{E}_{k}^s)$. We note that $\boldsymbol{\Theta}$ is included in $\mathbf{P}_k^{1,s}$ and $\mathbf{P}_k^{2,s}$ in \eqref{eq:E}.  We  
	keep those terms of \eqref{eq:E} that include $\boldsymbol{\Theta}$ and do simplifications as \vspace{-.2cm}
\begin{align}
\hspace{-1cm} \	\Pi(\boldsymbol{\Theta}) = \sum_{k=1}^{K}\sum_{s\in \mathcal{S}_k} & \!\begin{multlined}[t]
		\bigg[ 2 \mathfrak{Re}\left(\mathrm{Tr}(\boldsymbol{\Theta}\boldsymbol{\Phi}_s)\right) \\
		\hspace{-2cm}+ \mathrm{Tr}\bigg(\sum_{i=1}^{2} \Big[ \boldsymbol{\Theta}^H\boldsymbol{\Upsilon}_{hs,k}^{i,s} \boldsymbol{\Theta}\boldsymbol{\Gamma}_{hs}^{i,s} + \boldsymbol{\Theta}^T\boldsymbol{\Upsilon}_{ts,k}^{i,s}\noindent \boldsymbol{\Theta}\boldsymbol{\Gamma}_{ts}^{i,s} \\ \noindent
		\hspace{-3cm}+ \boldsymbol{\Theta}^H\boldsymbol{\Upsilon}_{hc,k}^{i,s} \boldsymbol{\Theta}^*\boldsymbol{\Gamma}_{hc}^{i,s} + \boldsymbol{\Theta}^T\boldsymbol{\Upsilon}_{tc,k}^{i,s} \boldsymbol{\Theta}^*\boldsymbol{\Gamma}_{tc}^{i,s} \Big] \bigg) \bigg],
	\end{multlined} \label{eq:obj2}
\end{align}
where 
\begin{align}
	\boldsymbol{\Upsilon}_{hs,k}^{1,s} &= \boldsymbol{\Upsilon}_{hs,k}^{2,s} = \mathbf{G}^{s^H} \mathbf{K}_{1}^H {\mathbf{U}_{k}^{s}} \mathbf{W}_{k}^s {\mathbf{U}_{k}^{s^H}} \mathbf{K}_{1} \mathbf{G}^{s}, \label{eq:upsilon1} \\
	\boldsymbol{\Upsilon}_{hc,k}^{1,s} &= \boldsymbol{\Upsilon}_{hc,k}^{2,s} = \mathbf{G}^{s^H} \mathbf{K}_{1}^H {\mathbf{U}_{k}^{s}} \mathbf{W}_{k}^s {\mathbf{U}_{k}^{s^H}} \mathbf{K}_{2} \mathbf{G}^{-s^*}, \label{eq:upsilon2} \\
	\boldsymbol{\Upsilon}_{ts,k}^{1,s} &= \boldsymbol{\Upsilon}_{ts,k}^{2,s} = \mathbf{G}^{-s^T} \mathbf{K}_{2}^H {\mathbf{U}_{k}^{s}} \mathbf{W}_{k}^s {\mathbf{U}_{k}^{s^H}} \mathbf{K}_{1} \mathbf{G}^{s}, \label{eq:upsilon3} \\
	\boldsymbol{\Upsilon}_{tc,k}^{1,s} &= \boldsymbol{\Upsilon}_{tc,k}^{2,s} = \mathbf{G}^{-s^T} \mathbf{K}_2^{H} \mathbf{U}_k^{s} \mathbf{W}_{k}^s \mathbf{U}_k^{s^H} \mathbf{K}_2 \mathbf{G}^{-s^*},\label{eq:upsilon4}
\end{align}
and  $\boldsymbol{\Gamma}_{hs}^{1,s}=\sum_{i=1}^{K}\mathbf{H}_{i}^{s} \mathbf{D}_{1i} \mathbf{V}_{i}^s \mathbf{V}_{i}^{s^H} \mathbf{D}_{1i}^{H} \mathbf{H}_{i}^{s^H}$, 
$\boldsymbol{\Gamma}_{hs}^{2,s}=\sum_{i=1}^{K}\mathbf{H}_{i}^{s} \mathbf{D}_{2i}^* \mathbf{V}_{i}^{-s^*} \mathbf{V}_{i}^{-s^T} \mathbf{D}_{2i}^{T} \mathbf{H}_{i}^{s^H}$, 
$\boldsymbol{\Gamma}_{hc}^{1,s}=\sum_{i=1}^{K} \mathbf{H}_{i}^{-s^*} \mathbf{D}_{2i} \mathbf{V}_{i}^s \mathbf{V}_{i}^{s H} \mathbf{D}_{1i}^{H} \mathbf{H}_{i}^{s^H}$, 
$\boldsymbol{\Gamma}_{hc}^{2,s}=\sum_{i=1}^{K} \mathbf{H}_{i}^{-s^*} \mathbf{D}_{1i}^* \mathbf{V}_{i}^{-s^*} \mathbf{V}_{i}^{-s^T} \mathbf{D}_{2i}^{T} \mathbf{H}_{i}^{s^H}$, 
$\boldsymbol{\Gamma}_{ts}^{1,s}=\sum_{i=1}^{K}  \mathbf{H}_{i}^{s} \mathbf{D}_{1k} \mathbf{V}_{k}^{s} \mathbf{V}_{k}^{s ^H} \mathbf{D}_{2k}^H \mathbf{H}_{k}^{-s^T}$, 
$\boldsymbol{\Gamma}_{ts}^{2,s}=\sum_{i=1}^{K}  \mathbf{H}_{i}^{s} \mathbf{D}_{2k}^* \mathbf{V}_{i}^{-s^*} \mathbf{V}_{i}^{-s^T} \mathbf{D}_{1k}^{T} \mathbf{H}_{i}^{-s^T}$, 
$\boldsymbol{\Gamma}_{tc}^{1,s}=\sum_{i=1}^{K} \mathbf{H}_{i}^{-s^*} \mathbf{D}_{2i} \mathbf{V}_{i}^s \mathbf{V}_{i}^{s^H} \mathbf{D}_{2i}^{H} \mathbf{H}_{i}^{-s^T}$, 
and $\boldsymbol{\Gamma}_{tc}^{2,s}=\sum_{i=1}^{K}  \mathbf{H}_{i}^{-s^*} \mathbf{D}_{1i}^* \mathbf{V}_{i}^{-s^*} \mathbf{V}_{i}^{-s^T} \mathbf{D}_{1i}^{T} \mathbf{H}_{i}^{-s^T}$.
	Furthermore, $\boldsymbol{\Phi}_s$ is as
	\begin{align}
		&\boldsymbol{\Phi}_s=  \mathbf{H}_{k}^{s} \mathbf{D}_{1k} \mathbf{V}_{k}^s \mathbf{V}_{k}^{s^H} 
		\mathbf{D}_{1k}^{H} \mathbf{R}_{k}^{s^H} \mathbf{K}_1^{H}  \mathbf{U}_{k}^{s}  \mathbf{W}_{k}^{s}  \mathbf{U}_{k}^{s^H} \mathbf{K}_1 \mathbf{G}^{s}\nonumber\\
		&+\mathbf{H}_{k}^{s} \mathbf{D}_{1k} \mathbf{V}_{k}^s \mathbf{V}_{k}^{s^H} 
		\mathbf{D}_{2k}^{H} \mathbf{R}_{k}^{-s^T} \mathbf{K}_2^{H}  \mathbf{U}_{k}^{s}  \mathbf{W}_{k}^{s}  \mathbf{U}_{k}^{s^H} \mathbf{K}_1 \mathbf{G}^{s}\nonumber\\
		&+\mathbf{G}^{-s^T} \mathbf{K}_2^{H} \mathbf{U}_{k}^{s} \mathbf{W}_{k}^{s}  \mathbf{U}_{k}^{s^H}
		\mathbf{K}_1 \mathbf{R}_{k}^{s} \mathbf{D}_{1k} \mathbf{V}_{k}^s \mathbf{V}_{k}^{s^H} \mathbf{D}_{2k}^{H}  \mathbf{H}_{k}^{-s^T}\nonumber\\
		&+\mathbf{G}^{-s^T} \mathbf{K}_2^{H} \mathbf{U}_{k}^{s} \mathbf{W}_{k}^{s}  \mathbf{U}_{k}^{s^H}
		\mathbf{K}_2 \mathbf{R}_{k}^{-s^*} \mathbf{D}_{2k} \mathbf{V}_{k}^s \mathbf{V}_{k}^{s^H} \mathbf{D}_{2k}^{H}  \mathbf{H}_{k}^{-s^T}\nonumber\\
		&+ \mathbf{H}_{k}^{s} \mathbf{D}_{2k}^{*} \mathbf{V}_{k}^{-s^*} \mathbf{V}_{k}^{-s^T} 
		\mathbf{D}_{2k}^{T} \mathbf{R}_{k}^{s^H} \mathbf{K}_1^{H}  \mathbf{U}_{k}^{s}  \mathbf{W}_{k}^{s}  \mathbf{U}_{k}^{s^H} \mathbf{K}_1 \mathbf{G}^{s}\nonumber\\
		&+ \mathbf{H}_{k}^{s} \mathbf{D}_{2k}^{*} \mathbf{V}_{k}^{-s^*} \mathbf{V}_{k}^{-s^T} 
		\mathbf{D}_{1k}^{T} \mathbf{R}_{k}^{-s^T} \mathbf{K}_2^{H}  \mathbf{U}_{k}^{s}  \mathbf{W}_{k}^{s}  \mathbf{U}_{k}^{s^H} \mathbf{K}_1 \mathbf{G}^{s}\nonumber\\
		&+ \mathbf{G}^{-s^T} \mathbf{K}_2^{H} \mathbf{U}_{k}^{s} \mathbf{W}_{k}^{s}  \mathbf{U}_{k}^{s^H}
		\mathbf{K}_1 \mathbf{R}_{k}^{s} \mathbf{D}_{2k}^* \mathbf{V}_{k}^{-s^*} \mathbf{V}_{k}^{-s^T} \mathbf{D}_{1k}^{T}  \mathbf{H}_{k}^{-s^T}\nonumber\\
		&+ \mathbf{G}^{-s^T} \mathbf{K}_2^{H} \mathbf{U}_{k}^{s} \mathbf{W}_{k}^{s}  \mathbf{U}_{k}^{s^H}
		\mathbf{K}_2 \mathbf{R}_{k}^{-s^*} \mathbf{D}_{1k}^* \mathbf{V}_{k}^{-s^*} \mathbf{V}_{k}^{-s^T} \mathbf{D}_{1k}^{T}  \mathbf{H}_{k}^{-s^T}
		\nonumber\\
		&- \mathbf{H}_{k}^{s} \mathbf{D}_{1k} \mathbf{V}_{k}^{s} \mathbf{W}_{k}^{s}  \mathbf{U}_{k}^{s^H} \mathbf{K}_1 \mathbf{G}^{s} -\mathbf{G}^{-s^T} \mathbf{K}_2^H\mathbf{U}_{k}^{s}\mathbf{W}_{k}^{s}\mathbf{V}_{k}^{s^H}\mathbf{D}_{2k}^H \mathbf{H}_{k}^{-s^T}.
		\nonumber
	\end{align}
	
	\begin{proposition}
		The subproblem with respect to $\boldsymbol{\Theta}$ is convex.
	\end{proposition}
	\begin{proof}
	The second derivative of \eqref{eq:obj2} is obtained as $\frac{\partial^2 \Pi(\boldsymbol{\Theta})}{\partial \boldsymbol{\Theta}^* \partial \boldsymbol{\Theta}}= \sum_{i=1}^{2}\left(\boldsymbol{\Upsilon}_{hs,k}^{i,s}\right)^T \left(\boldsymbol{\Gamma}_{hs}^{i,s}\right)^T + \left( \boldsymbol{\Gamma}_{tc}^{i,s} \right)^T \left( \boldsymbol{\Upsilon}_{tc,k}^{i,s}  \right)^T$. We note that the product of two Hermitian PSD matrices is also PSD. Additionally, the sum of PSD matrices is a PSD matrix. The constraint \eqref{eq:ampin} is convex because the complex modulus function is a convex function (being a norm), and the sublevel set of a convex function is always convex. Positive semi-definiteness of the above matrix and convexity of the constraint \eqref{eq:ampin} ensure the convexity with respect to $\boldsymbol{\Theta}$.
	\end{proof}
	
	If $\boldsymbol{\theta}$ is an arbitrary vector and $\boldsymbol{\Theta}=\text{diag}(\boldsymbol{\theta})$, we have two useful equalities $\text{Tr}(\boldsymbol{\Theta}^H\mathbf{A}\boldsymbol{\Theta} \mathbf{B})=\boldsymbol{\theta}^H(\mathbf{A}\odot\mathbf{B}^T)\boldsymbol{\theta}$ and $\text{Tr}(\boldsymbol{\Theta}\mathbf{A})=\text{diag}(\mathbf{A})^T\boldsymbol{\theta}$. We simplify \eqref{eq:obj2} using these two equations. The minimization with respect to $\boldsymbol{\theta}$ is as
% TODO: Replace the old verbose implementation with this refactored version.
\begin{align}\hspace{-.2cm}
	\min_{\boldsymbol{\theta}} \quad & \hspace{-.2cm}\sum_{k=1}^{K}\sum_{s\in \mathcal{S}_k}\sum_{i=1}^{2} \bigg[ \boldsymbol{\theta}^H(\boldsymbol{\Upsilon}_{hs,k}^{i,s}\odot\boldsymbol{\Gamma}_{hs}^{i,s^T})\boldsymbol{\theta} + \boldsymbol{\theta}^T(\boldsymbol{\Upsilon}_{ts,k}^{i,s}\odot\boldsymbol{\Gamma}_{ts}^{i,s^T})\boldsymbol{\theta} \nonumber \\
	& \hspace{.05cm} + \boldsymbol{\theta}^H(\boldsymbol{\Upsilon}_{hc,k}^{i,s}\odot\boldsymbol{\Gamma}_{hc}^{i,s^T})\boldsymbol{\theta}^* + \boldsymbol{\theta}^T(\boldsymbol{\Upsilon}_{tc,k}^{i,s}\odot\boldsymbol{\Gamma}_{tc}^{i,s^T})\boldsymbol{\theta}^* \bigg] \nonumber \\
	& + 2\mathfrak{Re}(\,\mathrm{diag}(\boldsymbol{\Phi}_s)^T\boldsymbol{\theta}) \label{opt:diag_compact} \\
	\textrm{s.t.} \quad & \eqref{eq:ampin}.\nonumber
\end{align}
Based on \eqref{eq:upsilon1}--\eqref{eq:upsilon4}, we define $\boldsymbol{\Sigma}_{\mu} = \sum_{k=1}^{K}\sum_{s\in \mathcal{S}_k}\boldsymbol{\Upsilon}_{\mu,k}^{2,s}\odot(\boldsymbol{\Gamma}_{\mu}^{1,s}+\boldsymbol{\Gamma}_{\mu}^{2,s})^T$ for $\mu \in \{hs, hc, ts, tc\}$. We rewrite $\boldsymbol{\theta}$ in the complex format as $\mathbf{a}^r + \iota \mathbf{a}^i, \mathbf{a}^r ,\mathbf{a}^i\in \mathbb{R}^{QM\times 1}$. 
	Then, \eqref{opt:diag_compact} becomes  
	\begin{subequations}\label{opt:xy}
		\begin{align}
			\max_{\mathbf{a}^r,\mathbf{a}^i} \quad & \hspace{0cm}(\mathbf{a}^r+\iota \mathbf{a}^i)^H\boldsymbol{\Sigma}_{hs}(\mathbf{a}^r+\iota\mathbf{a}^i)+(\mathbf{a}^r
			+\iota \mathbf{a}^i)^T\boldsymbol{\Sigma}_{ts}(\mathbf{a}^r+\iota\mathbf{a}^i)\nonumber\\
			&\hspace{-0.8cm}+(\mathbf{a}^r
			+\iota \mathbf{a}^i)^H\boldsymbol{\Sigma}_{hc}(\mathbf{a}^r+\iota\mathbf{a}^i)^*+(\mathbf{a}^r
			+\iota \mathbf{a}^i)^T\boldsymbol{\Sigma}_{tc}(\mathbf{a}^r+\iota\mathbf{a}^i)^*\nonumber\\
			& \hspace{-0.8cm} +2\mathfrak{Re}(\text{diag}(\boldsymbol{\Phi}_s)^T(\mathbf{a}^r+\iota \mathbf{a}^i))\label{eq:xy}\\
			\textrm{s.t.} \quad &   \sqrt{([\mathbf{a}^r]_{\ell})^2+  ([\mathbf{a}^i]_{\ell})^2} \leq 1, \hspace{.1cm}\forall \ell .\label{eq:tran}
		\end{align}
	\end{subequations}
Next, we propose an algorithm to solve \eqref{opt:xy} globally.

	\subsection{The Proposed BCD to Optimize RIS Elements}

	We reformulate \eqref{opt:xy}, a quadratically constrained quadratic program (QCQP),  in real-variable format for the purpose of algorithm development.  We define $\boldsymbol{\nu}=[\mathbf{a}^{r^T},\mathbf{a}^{i^T}]^T$. Based on this, it is easy to show that \eqref{eq:xy} can be equivalently written as $\boldsymbol{\nu}^T\boldsymbol{\Delta}\boldsymbol{\nu}+2\boldsymbol{\omega}^T\boldsymbol{\nu}$, where $\boldsymbol{\Delta}= \mathbf{T}(\boldsymbol{\Sigma}_{hs},1,-1,1,1)+\mathbf{T}(\boldsymbol{\Sigma}_{tc},1,1,-1,1)+\mathbf{T}(\boldsymbol{\Sigma}_{hc},1,1,1,-1)+\mathbf{T}(\boldsymbol{\Sigma}_{ts},1,-1,-1,-1)$, $\boldsymbol{\omega}=[\mathfrak{Re}(\text{diag}(\boldsymbol{\Phi}_s));-\mathfrak{Im}(\text{diag}(\boldsymbol{\Phi}_s))]$, where the operator $\mathbf{T}(\cdot)$ is defined as
	\begin{align}
		\mathbf{T}(\boldsymbol{\Sigma},e_1,e_2,e_3,e_4) = \begin{pmatrix}
			e_1\mathfrak{Re}(\boldsymbol{\Sigma}) & e_2\mathfrak{Im}(\boldsymbol{\Sigma})\\
			e_3\mathfrak{Im}(\boldsymbol{\Sigma}) & e_4\mathfrak{Re}(\boldsymbol{\Sigma})
		\end{pmatrix}.\nonumber
	\end{align}
	Moreover, \eqref{eq:tran} is rewritten as $\sqrt{[\boldsymbol{\nu}]_\ell^2+ [\boldsymbol{\nu}]_{\ell+QM}^2} \leq 1$, which is separable across RIS elements indices $\{\ell, \ell+QM\}$.  
		
The optimization problem is a QCQP, $\min_{\boldsymbol{\nu}} \boldsymbol{\nu}^T\boldsymbol{\Delta}\boldsymbol{\nu}+2\boldsymbol{\omega}^T\boldsymbol{\nu}$, with $QM$ separable unit-ball constraints of the form $\|[[\boldsymbol{\nu}]_\ell, [\boldsymbol{\nu}]_{\ell+QM}]^T\|_2 \leq 1$. This structure is ideal for a BCD algorithm that cyclically optimizes each block $\boldsymbol{\nu}_\ell = [[\boldsymbol{\nu}]_\ell, [\boldsymbol{\nu}]_{\ell+QM}]^T$ while keeping other variables fixed. For each block $\ell$, the subproblem is to solve $\min_{\boldsymbol{\nu}_\ell} \boldsymbol{\nu}_\ell^T\boldsymbol{\Delta}_\ell\boldsymbol{\nu}_\ell + 2\boldsymbol{\omega}_\ell^T\boldsymbol{\nu}_\ell$ subject to $\|\boldsymbol{\nu}_\ell\|_2 \leq 1$. Here, $\boldsymbol{\Delta}_\ell$ is the $2 \times 2$ submatrix with $[\boldsymbol{\Delta}_\ell]_{1,1} = [\boldsymbol{\Delta}]_{\ell,\ell}$, $[\boldsymbol{\Delta}_\ell]_{1,2} = [\boldsymbol{\Delta}]_{\ell,\ell+QM}$, $[\boldsymbol{\Delta}_\ell]_{2,1} = [\boldsymbol{\Delta}]_{\ell+QM,\ell}$, and $[\boldsymbol{\Delta}_\ell]_{2,2} = [\boldsymbol{\Delta}]_{\ell+QM,\ell+QM}$, and $[\boldsymbol{\omega}_\ell]_1 = \sum_{j \neq \ell, \ell+QM} [\boldsymbol{\Delta}]_{\ell,j} [\boldsymbol{\nu} ]_j + [\boldsymbol{\omega}]_\ell$ and $[\boldsymbol{\omega}_\ell]_2 = \sum_{j \neq \ell, \ell+QM} [\boldsymbol{\Delta}]_{\ell+QM,j} [\boldsymbol{\nu} ]_j+ [\boldsymbol{\omega}]_{\ell+QM}$ account for cross-coupling to fixed variables and the linear term. This 2D subproblem has a closed-form solution: if the unconstrained minimum $\boldsymbol{\nu}_\ell^* = -\boldsymbol{\Delta}_\ell^{-1}\boldsymbol{\omega}_\ell$ is feasible (i.e., $\|\boldsymbol{\nu}_\ell^*\|_2 \leq 1$), it is the solution. Otherwise, the solution is its projection onto the unit ball, $\boldsymbol{\nu}_\ell^*/\|\boldsymbol{\nu}_\ell^*\|_2$. We alternatively optimize different blocks $\ell$ until the convergence of this BCD method is attained. Due to the problem convexity, the proposed BCD is globally convergent \cite{bertsekas1999nonlinear}.

	\begin{figure*}
		\centering % Centers the figure
		% First figure
		\begin{minipage}[t!]{.3\textwidth} % [b] aligns at the bottom
			\includegraphics[width=\textwidth]{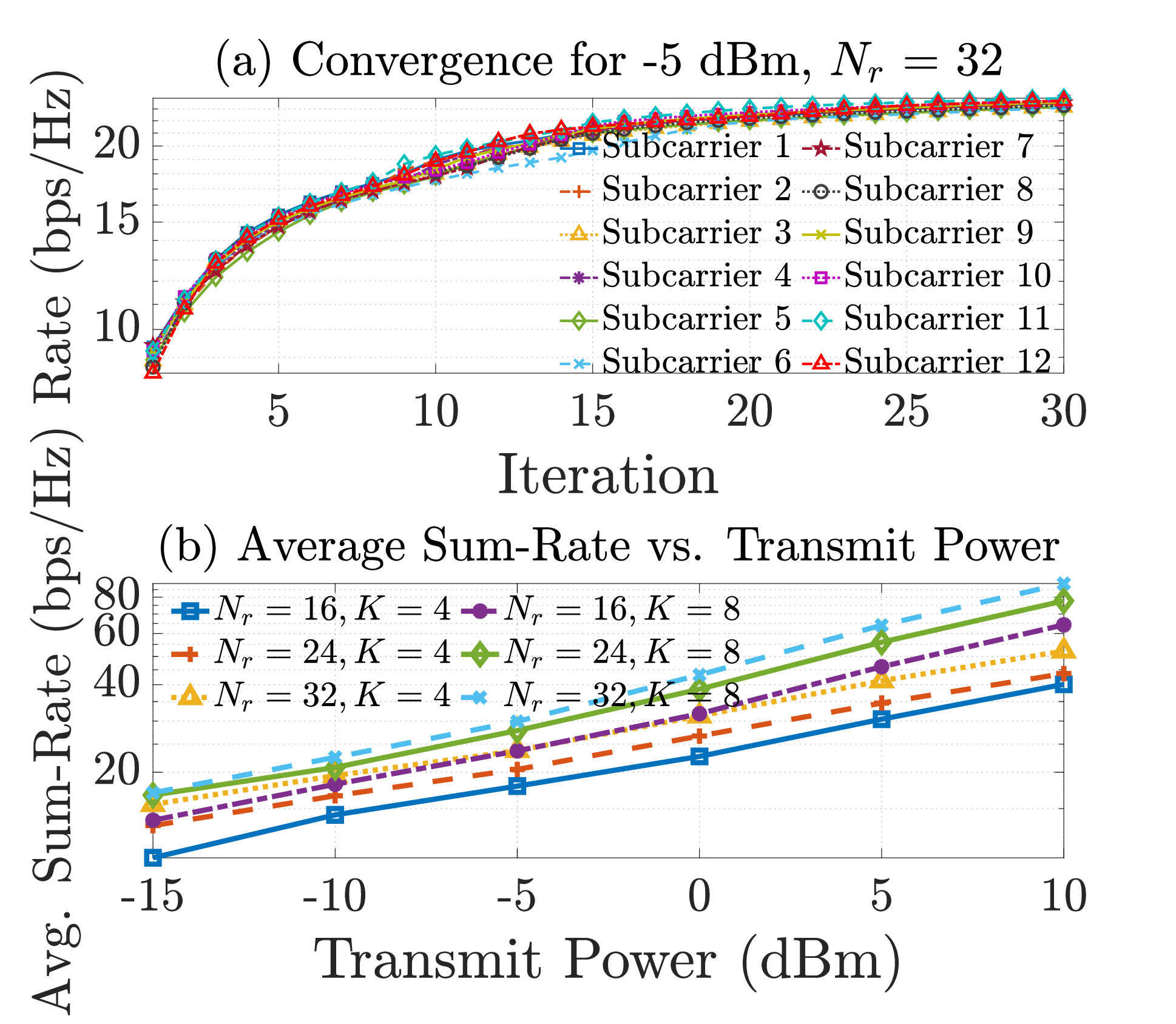}
			\caption{}
			\label{fig:conv}
		\end{minipage}
		\hfill % Space between the first and the second figures
		% Second figure
		\begin{minipage}[t!]{.3\textwidth} % [b] aligns at the bottom
			\includegraphics[width=\textwidth]{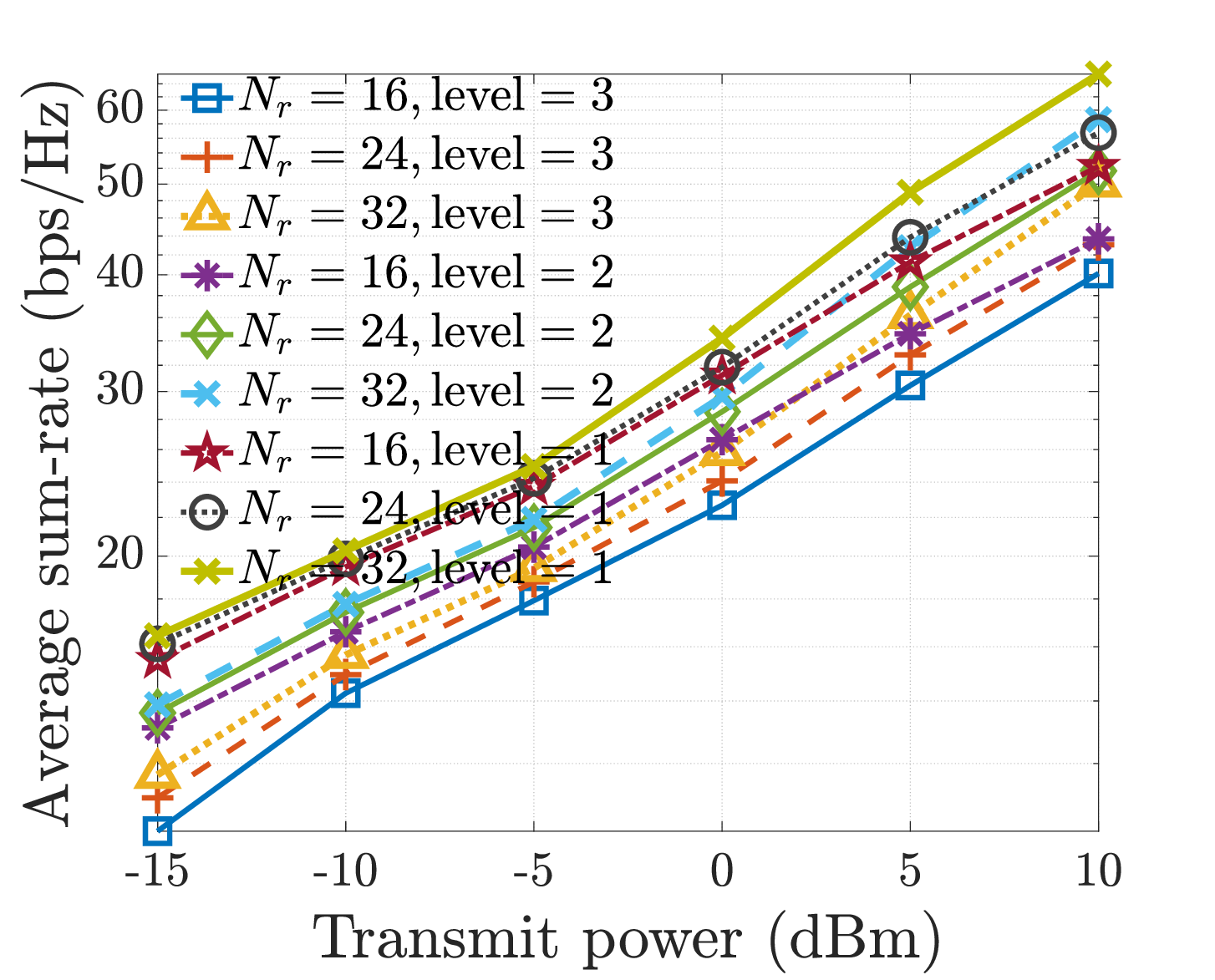}
			\caption{}
			\label{fig:sumrate}
		\end{minipage}
		\hfill % Space between the second and the third figures
		% Third figure
		\begin{minipage}[t!]{.3\textwidth} % [b] aligns at the bottom
			\includegraphics[width=\textwidth]{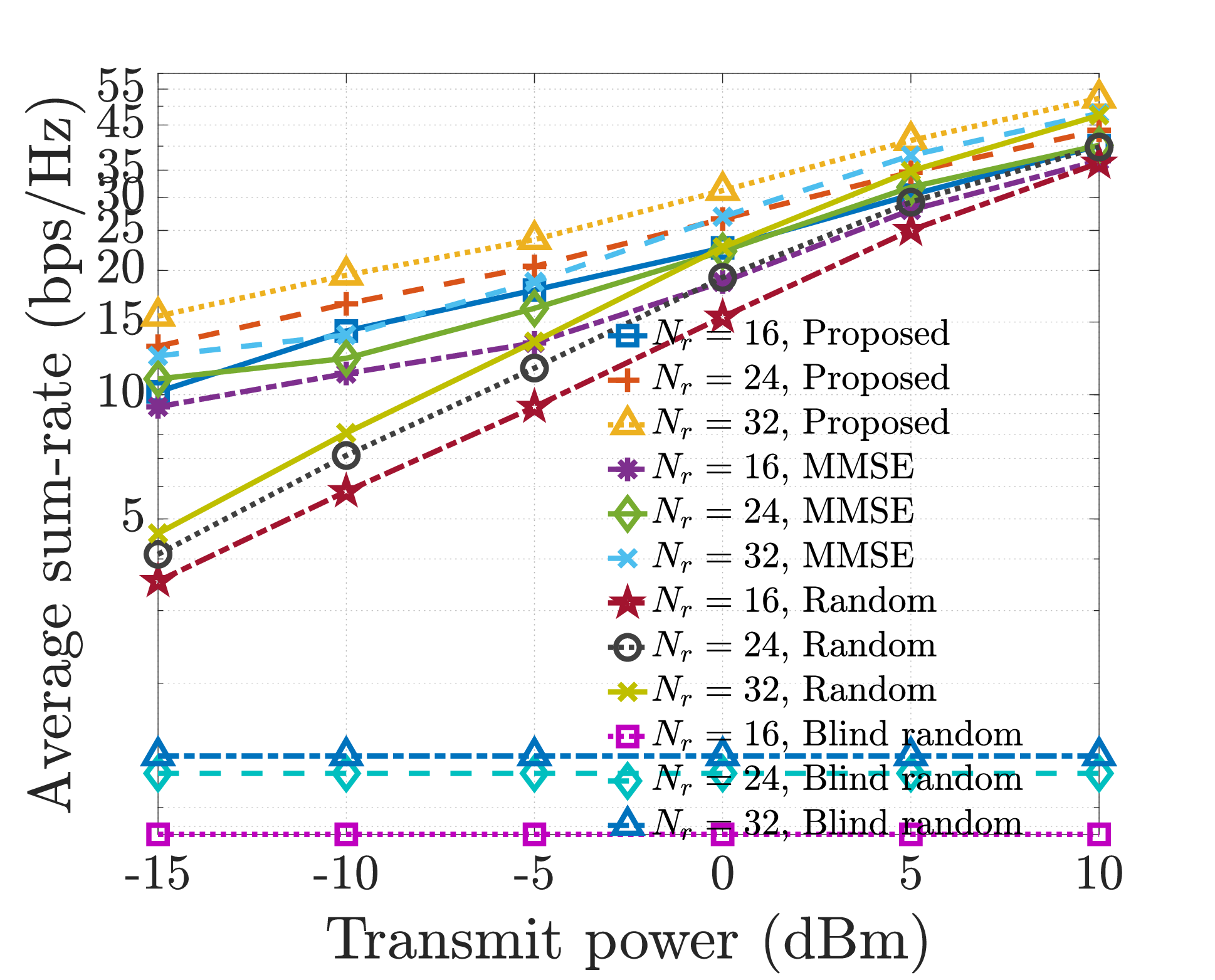}
			\caption{}\label{fig:figure3}
		\end{minipage}
		\caption*{Fig. 1a: The algorithm convergence. Fig. 1b: Average per-subcarrier sum-rate for $4$ and $8$ UEs. Fig. 2: The average per-subcarrier sum-rate with different IQI severities. Fig. 3: The average per-subcarrier sum-rate obtained by different methods.}
	\end{figure*}
	
	\subsection{The Proposed Overall Iterative BCD Approach} 
	Different blocks $\mathbf{U}_k^{s}$,  $\mathbf{W}_k^{s}$, $\mathbf{V}_k^{s}$, and $\boldsymbol{\Theta}$ are initialized. Sequentially, a block is chosen in a round-robin fashion to be updated, while other blocks are fixed. In the ($q$+1)$^{\text{th}}$ iteration of the proposed BCD method, a) $\mathbf{U}_k^{s,q+1}$ is updated with  $\{\mathbf{W}_k^{s,q},\mathbf{V}_k^{s,q},\boldsymbol{\Theta}^q\}$; b) $\mathbf{W}_k^{s,q+1}$ is updated with $\{\mathbf{U}_k^{s,q+1},\mathbf{V}_k^{s,q},\boldsymbol{\Theta}^q\}$; c) $\mathbf{V}_k^{s,q+1}$ is updated using \eqref{eq:V} with $\{\mathbf{U}_k^{s,q+1},\mathbf{W}_k^{s,q+1},\boldsymbol{\Theta}^q\}$; d) $\boldsymbol{\Theta}^{q+1}$ is updated using the proposed BCD with $\{\mathbf{U}_k^{s,q+1},\mathbf{W}_k^{s,q+1},\mathbf{V}_k^{s,q+1}\}$. We keep updating different blocks until all blocks converge. Due to convexity with respect to each block, the sequential updates of BCD converges to a stationary solution \cite[Prop. 3.7.1]{bertsekas1999nonlinear}.
	
	\section{Simulation Results}

	Four APs are at $(-30,-30,3)$, $(-30,30,3)$, $(30,-30,3)$, $(30,30,3)\,$m, two RISs are at $(-20,50,10)$, $(20,50,10)\,$m, and UEs are randomly distributed within a circle centered at $(0,350,1.5)\,$m with radius $30\,$m. The MIMO-OFDM system uses 12 subcarriers. The uplink channel from UE $k$ to the $j^{\text{th}}$ RIS on subcarrier $s$ is $\mathbf{H}_{j,k}^s = \sqrt{\frac{\kappa_{j,k}}{\kappa_{j,k}+1}} \sqrt{g_{j,k}} \bar{\mathbf{a}}_r(f^c ,\theta_{j,k}^{\text{AoA}}) \bar{\mathbf{a}}_t^H(f^c ,\theta_{j,k}^{\text{AoD}}) + \sum_{\ell=0}^{T-1} \sum_{p=1}^{P_\ell} \sqrt{\frac{g_{j,k,\ell,p}}{\kappa_{j,k}+1}}  \bar{\mathbf{a}}_r(f^c ,\theta_{j,k,\ell,p}^{\text{AoA}}) \bar{\mathbf{a}}_t^H(f^c ,\theta_{j,k,\ell,p}^{\text{AoD}}) h_{j,k,\ell,p}$\\$\times e^{-j2\pi\ell s/S}$, where $\bar{\mathbf{a}}_t(\cdot)$ and $\bar{\mathbf{a}}_r(\cdot)$ represent the transmit and receive array response vectors, respectively, and $f^c$ is the carrier frequency. Other channels follow the same Rician model. The UE and AP array responses follow uniform linear array (ULA) models, while the RIS response follows a uniform planar array (UPA) model, all with half-wavelength spacing. Here, $\theta_{j,k}^{\text{AoD}}$ and $\theta_{j,k}^{\text{AoA}}$ are the angles of departure and arrival, respectively. For NLOS components, the angles are randomly distributed around the main LOS angles. Per channel tap $\ell$, we have $P_\ell=4$ paths. The system has $M=64$ RIS elements, $h_{j,k,\ell,p} \sim \mathcal{CN}(0,1)$, $T=4$, and Rician factor $\kappa_{j,k} = 10$. LOS path loss follows 3GPP as $
	PL_{\text{LOS}}^{\text{dB}} = 22 \log_{10}(d_k) + 28 + 20 \log_{10}(f^c)+ \xi_{\text{LOS}},$
	where $d_k$ is the distance in meters, $f^c$ is the carrier frequency in GHz, with $f^c = 28$ GHz, $\xi_{\text{LOS}} \sim \mathcal{N}(0, 5.8^2)$, $g_{j,k} = 10^{-PL_{\text{LOS}}^{\text{dB}}/10}$, $g_{j,k,\ell,p} = 10^{-PL_{\text{NLOS}}^{\text{dB}}/10}$ \cite{3GPP_TR_36_814_2017}. NLOS path loss is derived similarly with $\xi_{\text{NLOS}} \sim \mathcal{N}(0, 8^2)$. The system is configured with a bandwidth of $B = 180$~kHz with a spacing of $\Delta f = 15$~kHz. We investigate three distinct levels of IQI, where parameters are drawn from uniform distributions. For level~1, the amplitude and phase imbalance ranges are $\mathcal{U}[0.9, 1.1]$ and $\mathcal{U}[-10^\circ, 10^\circ]$; for level~2, the ranges are $\mathcal{U}[0.8, 1.2]$ and $\mathcal{U}[-20^\circ, 20^\circ]$; and for level~3, they are $\mathcal{U}[0.7, 1.3]$ and $\mathcal{U}[-30^\circ, 30^\circ]$.
	
 In Fig. \ref{fig:conv}a, we illustrate the convergence of the proposed algorithm with four UEs and level 3 IQI.	Notably, the algorithm converges after approximately 30 iterations. One hundred fifty random channels and IQI instances are generated and the sum-rate is averaged over different realizations. With level 3 IQI, the average achievable per-subcarrier sum-rate is shown in Fig. \ref{fig:conv}b. The simulation varies the maximum transmit power of each UE from $-15$ dBm to $10$ dBm in steps of $5$ dBm, considering both four and eight UEs. Furthermore, the simulation explores the impact of increasing the number of receive antennas per AP from $16$ to $32$. Notably, the average achievable rate increases with both more antennas and higher transmission power.
	In our simulations, we observed that, after the proposed algorithm converges, $|[\boldsymbol{\Theta}^j]_{m,m}|$ tends toward $1$. This observation of the algorithm's behavior is noteworthy because directly imposing $|[\boldsymbol{\Theta}^j]_{m,m}|=1$ in the problem formulation makes it more difficult to solve. 
	
	With 4 UEs and varying numbers of antennas, Fig. \ref{fig:sumrate} shows the average per-subcarrier sum-rate. The average per-subcarrier sum-rate decreases with increasing IQI severity. Fig. \ref{fig:figure3} evaluates the proposed joint optimization, referred to as ``Proposed'', by comparing its performance under level 3 IQI against three benchmark approaches. The first approach, ``MMSE'', sets $\mathbf{W}_k^s= \mathbf{I}_{b_k\times b_k}$, reducing the proposed method to MMSE precoding/combining. The second approach, ``Random'', is IQI-aware but uses randomly assigned, fixed RIS element values. The third method, ``Blind'', applies the conventional uplink WMMSE approach while neglecting IQI effects at both UEs and APs with random fixed RIS element values. From Fig. \ref{fig:figure3}, we observe that the proposed scheme results in a higher per-subcarrier average sum-rate, particularly at lower transmit powers.

	\section{Concluding Remarks}
	Joint precoding-combining design and multi-RIS reflection coefficient optimization for the uplink of CF-mMIMO-OFDM systems in the presence of IQI at UEs and APs were studied where uplink data is transmitted through a direct link and multiple RISs. The UE sum-rate maximization was formulated and a WMMSE-based BCD method was proposed to iteratively solve a sequence of convex problems, where efficient, low-cost methods were proposed to solve BCD subproblems globally. Extensive simulation results were given to demonstrate the proposed approaches performance.

	\ifCLASSOPTIONcaptionsoff
	\newpage
	\fi

	\bibliographystyle{IEEEbib}
	\bibliography{ref_SDRA}
	\newpage

\end{document}